\documentclass[12pt]{article}

\usepackage{epsfig}
\usepackage{amssymb}
\usepackage{amsfonts}

\usepackage{color}

\def\be{\begin{equation}}
\def\ee{\end{equation}}
\def\ba{\begin{array}{c}}
\def\ea{\end{array}}

\newcommand{\bea}{\begin{eqnarray}}
\newcommand{\eea}{\end{eqnarray}}

\newcommand{\kt}{\rangle}

\newtheorem{thm}{Theorem}
\newtheorem{cor}[thm]{Corollary}
\newtheorem{lemma}[thm]{Lemma}

\newenvironment{proof}{\noindent {\bf Proof}}{\hfill$\square$\vspace{3mm}\endtrivlist}

\begin{document}

\titlepage

 \begin{center}{\Large \bf

Asymptotic non-Hermitian
degeneracy phenomenon and its
exactly solvable simulation

  }\end{center}

 \begin{center}

\vspace{6mm}

  {\bf Miloslav Znojil} $^{1,2,3}$

\end{center}

\vspace{6mm}

  $^{1}$
 The Czech Academy of Sciences,
 Nuclear Physics Institute,
 \v{R}e\v{z} 292,
250 68 Husinec, Czech Republic, {e-mail: znojil@ujf.cas.cz}

 $^{2}$
 Department of Physics, Faculty of
Science, University of Hradec Kr\'{a}lov\'{e}, Rokitansk\'{e}ho 62,
50003 Hradec Kr\'{a}lov\'{e},
 Czech Republic

%
%

$^3$
School for Data Science and Computational Thinking, Stellenbosch
University, 7600 Stellenbosch,
 South Africa


\newpage

\subsection*{Abstract}

A conceptually
consistent \textcolor{black}{understanding is sought} for the
interactions sampled by the
imaginary cubic oscillator \textcolor{black}{with}
potential $V^{(ICO)}(x)={\rm i}x^3$
which
is, by itself, not acceptable \textcolor{black}{as a meaningful quantum model}
due to a combination of its
non-Hermiticity, unboundedness and, first of all,
due to the Riesz-basis non-diagonalizability
of the Hamiltonian
known as its
intrinsic exceptional point (IEP) feature.
For the purposes of a
\textcolor{black}{perturbation-theory-based}
simulation of
the emergence
of
such a singular \textcolor{black}{system},
a simplified
\textcolor{black}{(though not too strictly related)}
toy-model
Hamiltonian is proposed.
It combines
an
$N-$point
discretization of the real line of coordinates
with an {\it ad hoc\,}
interaction in a two-parametric
$N$-by-$N$-matrix
Hamiltonian $H=H^{(N)}(A,B)$.
After such a simplification,
\textcolor{black}{one
can still encounter a somewhat weaker
form of
non-diagonalizability
at} the \textcolor{black}{conventional}
Kato's exceptional-point (EP)
\textcolor{black}{limit of parameters}
$(A,B) \to (A^{(EP)},B^{(EP)})$.
The IEP-non-diagonalizability phenomenon itself
appears
mimicked
by the less enigmatic
EP degeneracy
of the discrete toy-model,
\textcolor{black}{especially at
large $N \gg 1$}.
\textcolor{black}{What we gain is that the}
regularization  \textcolor{black}{of
the simplified toy-model
in a vicinity}
of the \textcolor{black}{conventional
EP
becomes, in contrast to the IEP case,} feasible.


\subsection*{Keywords:}

non-Hermitian quantum physics of degeneracy phenomena;
degeneracies realized as the Kato's exceptional points (EPs);
construction of EPs in a solvable $N$ by $N$ matrix model;
mechanism of removal of the degeneracy via an {\em ad hoc\,} perturbation;
\textcolor{black}{re-confirmation of the
unphysical intrinsic-EP status} of the imaginary cubic oscillator

 \newpage

\section{Introduction\label{introduct}}

The concept of asymptotic non-Hermitian
degeneracy called
intrinsic exceptional point (IEP)
was introduced
by Siegl and Krej\v{c}i\v{r}\'{\i}k \cite{Siegl}.
They \textcolor{black}{detected the presence of} this characteristics
in the imaginary
cubic oscillator (ICO)
operator
 \be
 H^{(ICO)}= -\frac{d^2}{dx^2}+{\rm i}x^3\,.
 \label{icos}
 \ee
They proved that its eigenstates $|\psi^{(ICO)}_n\kt$
do not form Riesz basis,
so they concluded
that
such an operator cannot be considered diagonalizable
and that
``there
is no quantum-mechanical Hamiltonian
associated with it''~\cite{Siegl}.
In this context
we reconfirmed,
in  \cite{foundations}, that
the applicability
of quantum theory
at the non-Hermitian IEP dynamical extreme
\textcolor{black}{is questionable}.
We found that
many sophisticated construction
techniques which still do work
\textcolor{black}{for finite matrices}
cannot be transferred
to the
IEP dynamical regime.

A brief explanation is that near an $N$-by-$N$-matrix
singularity \textcolor{black}{with $N < \infty$},
the ``corridor of unitarity''  (cf. \cite{corridors})
becomes,
in the $N \to \infty$ limit, unacceptably
(i.e., exponentially) narrow.
As a consequence,
we felt forced to conclude that ``the
practical realizations of the standard quantum-mechanical
ICO model
remain \ldots
elusive'' \cite{foundations},
and that
``the currently unresolved status  \ldots
of the IEP-related
instabilities
does not seem to have an easy resolution'' \cite{foundations}.

In
our present paper we
intend to propose another (partial) \textcolor{black}{clarification
of the puzzle}.
The not quite expected existence of such a \textcolor{black}{new
explanation of the problem}
has two roots. First,
we imagined that
in the ICO model of Eq.~(\ref{icos})
as well as in its multiple IEP-singular analogues,
a regularization
of their IEP singularity
cannot proceed directly, via
an immediate perturbative regularization of the IEP-singular
differential operators themselves.
In a ``preparatory-step'' of our proposed
treatment
of the problem
\textcolor{black}{we decided to
mimic the relevant features od $H^{(ICO)}$ via}
certain limits of its discrete-coordinate $N$ by $N$ matrix
alternatives.
\textcolor{black}{For this purpose we are going to consider certain}
$N$-by-$N$-matrix Hamiltonians
of the form $H^{(N)}=T^{(N)} + V^{(N)}$
in which the
kinetic energy is represented by the
discrete Laplacean,
 \be
 T^{(N)}
  =
 \left[ \begin {array}{ccccc}
 2&-1&0&\ldots&0\\
 \noalign{\medskip}-1&2&-1&\ddots&\vdots\\
 \noalign{\medskip}0&\ddots&\ddots &\ddots&0\\
 \noalign{\medskip}\vdots&\ddots&-1& 2&-1\\
  \noalign{\medskip}0&\ldots&0&-1&  2
  \end {array} \right]\,.
 \label{Ka8pV}
 \ee
Such an operator is, for any local and real potential $V^{(N)}=V^{(N)}(x)$, Hermitian.
Moreover,
\textcolor{black}{its kinetic-energy component (\ref{Ka8pV})
can be interpreted as an immediate
discrete analogue of its continuous-coordinate partner $T=-d^2/dx^2$
(cf., e.g., the recent preprint \cite{Dusan}
for details and/or further references;
here we use units such that $\hbar=1$ and mass $m=1/2$)}.

The second root and step of our present
IEP-simulation proposal
was inspired by our other
paper \cite{jmp1}.
In a way proposed in this paper,
Hamiltonians \textcolor{black}{$H^{(N)}=T^{(N)} + V^{(N)}$}
can be also perceived as the
discrete-coordinate analogues
of \textcolor{black}{certain more general non-Hermitian}
differential-operator partners as sampled by the ICO model of
Eq.~(\ref{icos}): The correspondence may be mediated, e.g., by the
evaluations of a suitable given potential at the grid-point
coordinates $x_1$, $x_2$, \ldots, $x_N$
(for more details see section \ref{primark} below).
\textcolor{black}{This may be expected to simplify the analysis.
Thus, in \cite{jmp1} we} studied these approximants, with
the basic message being,
from our present point of view, just
a mixture of good and bad news
(see also
a \textcolor{black}{few further related comments below}).

What remained encouraging was
an observation
that
the construction
of the necessary
$N$ by $N$ matrix
approximants
might be a feasible task,
especially when one admits
an ample use of
computer-assisted symbolic manipulations.
In parallel,
a new discouragement emerged
when we noticed that
with the growth of $N$,
the memory-and-time
costs of the
construction appeared to
increase so quickly
that we were only able to
test its performance
up to comparatively small
$N=6$.
Thus, we
had to
conclude, temporarily,
that the desirable ``extrapolation
to any hypothetical continuous-coordinate limit $N \to \infty$
does not seem to be feasible at present'' \cite{jmp1}.

Very recently, our skepticism faded away
when we imagined that
for a
simulation of the IEP degeneracy
it is in fact not necessary for us to insist on the
\textcolor{black}{entirely general form of
the family of the discrete potentials $V^{(N)}(x_i)$, especially because the}
IEP degeneracy itself
is merely an asymptotic, higher-excitation phenomenon.
The \textcolor{black}{consequences of the}
growth of our integer parameter $N$
\textcolor{black}{could equally well be studied, therefore,
using any suitable toy-model $V^{(N)}(x_i)$.}

This idea immediately led to an innovated formulation of our
``partial-remedy''
project (cf. section \ref{upatabe}).
The essence of the innovation was that
for the purposes of the IEP-degeneracy simulation
during the continuous-coordinate limit $N \to \infty$,
the \textcolor{black}{structure of the degeneracy}
itself may and will be kept elementary and fixed.
This led to the
results which will be
described in sections
\ref{drittark}
and~\ref{patabe},
with a compact summary added
in section \ref{finis}.


\section{\textcolor{black}{Difference-operator Hamiltonians } \label{primark}}

Although the IEP-singular ICO operator (\ref{icos})
cannot serve as
quantum Hamiltonian,
its eigenstates
form a complete
set \cite{Siegl}
so that one feels tempted
to
regularize the model
using a suitable perturbation.
Along such a path of considerations
(see, e.g., \cite{foundations})
one almost immediately discovers that
besides the IEP-related asymptotic degeneracy (\ref{aprode}),
there emerges also another and, possibly,
equally serious technical obstacle
which
is related to the unboundedness of
the operator.
This is a challenge which
was already known to Dieudonn\'{e}
(cf.  \cite{Dieudonne}).
One of its
resolutions was recommended and described in review
\cite{Geyer} (cf. also \cite{Fabio}).
\textcolor{black}{As we already mentioned,}
a not too dissimilar
trick
based on the discretization of coordinates
will be also used in what follows.

\subsection{\textcolor{black}{Discrete large$-N$ version of}
square well}

%
%
\begin{figure}[h]                    
\begin{center}                         
\epsfig{file=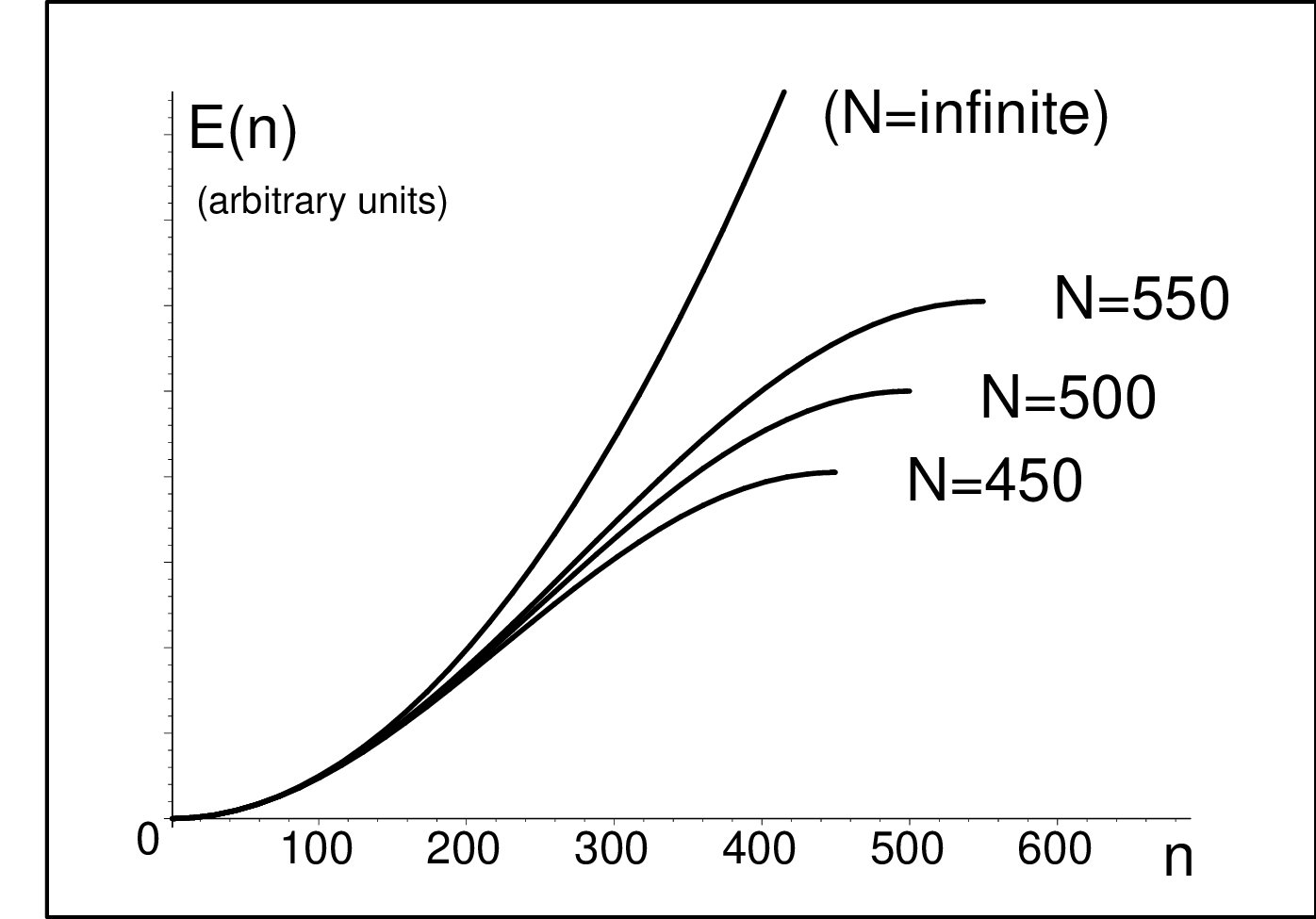,angle=270,width=0.35\textwidth}
\end{center}    
\caption{Curves $E(n)$ fitting the discrete-square-well spectra
(\ref{fofo}) of Schr\"{o}dinger equation living on an equidistant
grid-point lattice of a fixed length and variable \textcolor{black}{mesh-size}
$\lambda$.
 \label{globe}}
\end{figure}

%


%
\begin{figure}[h]                    
\begin{center}                         
\epsfig{file=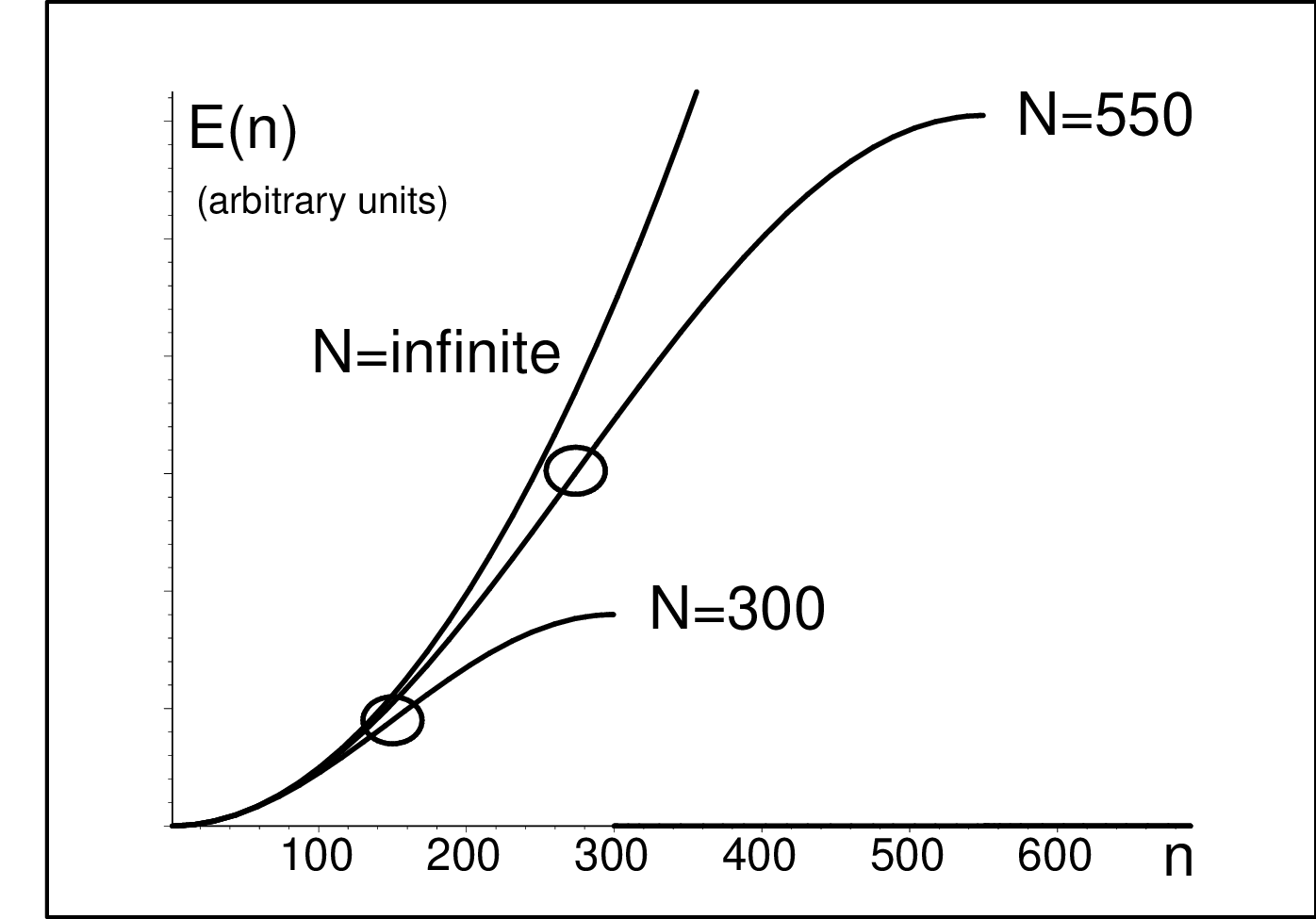,angle=270,width=0.35\textwidth}
\end{center}    
\caption{Same as Figure \ref{globe}, with the small circles marking the
central level. Beyond this ``privileged'' level,
the $N < \infty$ approximation
of the continuous-coordinate deep-square-well spectrum
ceases to be meaningful.
\textcolor{black}{At any sufficiently large $N < \infty$,
the central level can be perceived as
mimicking a typical ``highly excited'' state
of the $N=\infty$ system with $n \gg 1$.}
 \label{eglobe}}
\end{figure}

   \noindent
At a fixed  $N <
\infty$ \textcolor{black}{in
Schr\"{o}dinger equation
 \be
 \left [
 T^{(N)} + V^{(N)}
 \right ]|\psi_n\kt =  E(n) \, |\psi_n\kt\,
 \ \ \ \ n=1,2,\ldots
 \label{epro}
 \ee
}we may
\textcolor{black}{start our considerations by letting the
potential,
inside a finite interval of
a fixed length $L$, vanish, $V^{(N)}=V^{(N)}_{square\,well}=0$.
Then, the}
energy spectrum becomes strictly
positive and can be written
in closed form.
\textcolor{black}{Interested readers may
find the details in
the recent preprint
\cite{Dusan} -- for our present purposes,
it is sufficient to recall just the ultimate
energy-specifying formula
Nr. 3.15 of {\it loc. cit.}, viz.,}
 \be
 E^{(square\,well)}(n) =
 \left \{ \frac{ \sin [\pi\,n\,\lambda/L
 ]}{\lambda}
 \right \}^2\,,\ \ \ n = 1, 2, \ldots\,.
  \label{fofo}
 \ee
\textcolor{black}{The real width
of the infinitely deep square well
is equal here to the product $L=N\,\lambda$
of $N$ with the distance  $\lambda$
between the grid points.
This means that the decrease of $\lambda \to 0$
is equivalent to the growth of $N \to \infty$.}


Two aspects of the latter formula are relevant.
First, the
formula reproduces the continuous-coordinate square-well
spectrum  in the limit of $N \to \infty$
(see Figure \ref{globe}).
Second, the {\em whole}
continuous-coordinate square-well spectrum can be perceived as a
limit of the mere lower part of the respective $N < \infty$
spectrum (see Figure \ref{eglobe}).

\subsection{Inclusion of non-Hermitian potentials}

One of the characteristic features of the
square-well spectrum with $N=\infty$ is the steady growth of the
energy-level differences
${E_{n+1}^{(sqw)}}-{E_n^{(sqw)}}$ with $n$.
The inspection of the pictures reveals that at any finite
$N < \infty$,  such a feature is inadvertently lost beyond
$n=n_{\max}$ where $n_{\max}  =[ N/2]$. Thus, once one
decides to
\textcolor{black}{take a nontrivial $V^{(N)} \neq 0$
and }
study the spectra at the large $N\gg 1$, one can still safely
ignore the excited states with $n>n_{\max}$
as irrelevant.

For our present purposes it is important to know that all of
the latter observations
remain applicable especially after the inclusion of
small perturbations $V^{(N)} \neq 0$.
One should only keep in mind that
\textcolor{black}{even then}, the
spectrum need not be positive definite.
In this sense the above-mentioned \textcolor{black}{idea of the}
asymptotic irrelevance
of the upper half of the spectrum
of the discrete-coordinate toy models with
finite $N < \infty$
\textcolor{black}{must be properly modified}.

The elusive IEP degeneracy can be now
perceived as paralleled and mimicked by its
finite$-N$ simulation mediated by
the Kato's degeneracy-causing exceptional points
(EP, \cite{Kato}) or, more precisely,
by the Kato's exceptional points
of order $M$ (EPM). In this direction
we already pointed out, in \cite{foundations}, that
the ICO-related singularity finds its
analogue in the behavior of eigenstates $|\psi_n\kt$ of a
parameter-dependent non-Hermitian matrix
$H^{(N)}=H^{(N)}(\kappa)$ in the vicinity of
its exceptional-point singularity
$H^{(N)}(\kappa^{(EPM)})$ at a suitable
\textcolor{black}{degree of degeneracy}
$M\geq 2$.

Some of these observations were already developed in
paper \cite{jmp1}. We revealed there that a consequent and
systematic analysis of the multiparametric
and ${\cal PT}-$symmetric
matrix models
$H^{(N)}=T^{(N)}+V^{(N)}$ of the form
depending on a real $[N/2]-$plet of parameters, viz.,
 \be
 H^{(N)}=
 H^{(N)}(A,B,C,\ldots) =
 \left[ \begin {array}{ccccccc}
  -iA&-1&0
 &\ldots&&\ldots&0
 \\
 {}-1&-iB&-1&0&\ldots&\ldots&0
 \\
 {}0&-1&-iC&-1&\ddots&
 &\vdots
 \\
 {}0&0&\ddots&\ddots&\ddots&0&0
 \\
 {}\vdots&&\ddots&-1&iC&-1
 &0
 \\
 {}0&\ldots&\ldots&0&-1&iB&-1
 \\
 {}0&\ldots&&\ldots&0&-1&iA
 \end {array} \right]\,
 \label{ba0t}
 \ee
ceases to be easy even in the first
nontrivial three-parametric case with  $N=6$ and $M=6$.
On these grounds we decided
to circumvent the obstacles via
a
weakening of the strength of the assumptions. In place of the
over-ambitious search for the EPMs ``with a sufficiently large
$M$'', we will consider a reduced task in which the number of the free
parameters (determining also the
maximal order $M$ of the EPM degeneracies)
will be kept fixed and restricted to the first
few smallest integers.

\subsection{\textcolor{black}{Purely imaginary discrete potentials}}

In any ordinary differential
Schr\"{o}dinger equation $H\,\psi_n=E_n\,\psi_n$
and for
an arbitrary unbounded Hamiltonian
$H=-d^2/dx^2+V(x)$,
as we already pointed out,
the real line of coordinates $x \in \mathbb{R}$
(or, possibly, its finite segment) 
may be replaced by a discrete and equidistant
(and, say, finite)
grid-point lattice
$\{x_1,x_2,\ldots,x_N\}$.
In such a setting, a return
to the continuous-coordinate limit
can be mediated by the growth of $N \to \infty$
in combination with a
simultaneous decrease of the grid-point distance $\lambda \sim 1/N$.

For the special, ICO-inspired
class of models of our present interest,
the continuous-coordinate potentials
exhibiting the parity-time symmetry $H{\cal PT}={\cal PT}H$
will be assumed
purely imaginary.
Then,
the bounded-operator avatar of the initial Hamiltonian with such a
symmetry
acquires the matrix form of Eq.~(\ref{ba0t}) exhibiting the same
antilinear symmetry.

The knowledge of the parameters
might enable us to reconstruct
\textcolor{black}{(or, better,
interpolate/approximate)}
the continuous-coordinate potential
$V(x)$ in the limit of large $N \to \infty$.
For the sake of definiteness, we will assume that
all of the constants $A$, $B$, \ldots are real.
At a fixed $N$
our non-Hermitian but ${\cal PT}$ symmetric quantum Hamiltonian will then
have the $N$ by $N$ matrix form (\ref{ba0t}). Its spectrum
will be real, discrete and non-degenerate in an $N-$dependent domain
${\cal D}$ of
parameters which may be called ``physical''.
This domain is
not empty since it contains a subset of the
parameters which remain sufficiently small.
Model (\ref{ba0t}) can be then perceived as a perturbation of the
conventional square well with Hamiltonian
$H^{(N)}_{(sqw)}={H^{(N)}}(0,0,\ldots) $
represented by a matrix which is Hermitian, $H^{(N)}_{(sqw)}= \left [
H^{(N)}_{(sqw)}\right ]^\dagger$.

\subsection{Differential-difference-operator correspondence}

Incidentally, model (\ref{icos})
played, for years, the role of an
important benchmark example in
several branches of physics \cite{Fisher,DB,BB,Carl}.
The disproof of its probabilistic quantum-mechanical
tractability
was, therefore, disturbing.
Its unacceptability
has also been reconfirmed by G\"{u}nther with
Stefani \cite{Uwe} who provided a ``clear complementary evidence''
that the models of such a type ``are not equivalent to Hermitian
models'', mainly due to the IEP property. They called this property
``non-Rieszian mode behavior'' and, in a way which
inspired also our present constructive
considerations, they concluded that
``what is still lacking is a simple physical explanation
scheme for the non-Rieszian behavior of the
eigenfunction sets'' \cite{Uwe}.

We found it important that the
latter IEP-degeneracy property
of the eigenstates
$|\psi^{(ICO)}_n\kt$ of $H^{(ICO)}$
can be visualized, roughly speaking,
as a steady weakening of their mutual linear independence
with the growth of excitation,
 \be
 |\psi^{(ICO)}_n\kt \approx |\psi^{(ICO)}_{n+1}\kt
\,,\ \ \ \ n \gg 1 \,.
 \label{aprode}
 \ee
A deeper insight 
is obtained when
one recalls the notion
of exceptional point (EP) as introduced in
the classical Kato's
monograph on perturbation theory \cite{Kato}.
In spite of the fact that the Kato's attention has only
been restricted
to the finite and
parameter-dependent
$N$ by $N$ matrices
$H^{(N)}(\kappa)$,
his concept of the EP degeneracy
can really be found analogous
to its ``intrinsic'', asymptotic
version of Eq.~(\ref{aprode}).

The analogy
\textcolor{black}{is imperfect of course.
It only finds a partial support in the fact}
that
at a finite $N<\infty$,
the EP
can be also defined
as an instant of parallelization
of several (i.e., in general, of $M$)
eigenvectors
$|\psi^{(N)}_n(\kappa)\kt$
of $H^{(N)}(\kappa)$ at
$\kappa = \kappa^{(EPM)}$
with suitable $M \leq N$.
In comparison with IEP,
the \textcolor{black}{other} formal difference is that in the
EP limit (i.e., in the EPM limit)
$\kappa \to \kappa^{(EPM)}$,
the Kato's finite$-N$ degeneracy
involves not only an $M-$plet of
certain properly normalized eigenvectors,
 \be
 \lim_{\kappa \to \kappa^{(EPM)}}|\psi_{m_j}^{(N)}(\kappa)\kt=
 |\psi_{(EPM)}^{(N)}\kt\,,
 \ \ \ \
 j=1,2,\ldots,M
 \,,
 \label{wadeg}
 \ee
but also
the related eigenvalues,
 \be
 \lim_{\kappa \to \kappa^{(EPM)}}E_{m_j}^{(N)}(\kappa)=
 E_{(EPM)}^{(N)}\,,
 \ \ \ \
 j=1,2,\ldots,M
 \,.
 \label{endeg}
 \ee
In  \cite{foundations} we found the
``simplification'' provided by the absence of the
energy-degeneracy (\ref{endeg})
to be more than compensated,
in the
IEP context, by the technical
complications arising from the
``non-Rieszian mode behavior''~(\ref{aprode}).
Siegl with Krej\v{c}i\v{r}\'{\i}k
emphasized,
for this reason, that
the IEP
singularity is ``much stronger than any EP
associated with finite Jordan blocks'' \cite{Siegl}.
\textcolor{black}{Thus, our present tentative
replacement of the IEP-singular systems
by their discrete EPM analogues
could help us to understand at least some of the
deeper mathematical roots of the
rather subtle IEP-unacceptability
enigma}.

\section{The Kato's exceptional points at finite $N$\label{upatabe}}


\subsection{The domains of unitarity}

In our present paper we intend to
study the regularization of the IEP models
based on a simulation of their properties
using matrices  (\ref{ba0t}).
Our motivation is \textcolor{black}{threefold}.
Firstly, we feel impressed by the
empirical observation that all of these matrices
share
certain not quite expected ``exact solvability'' features.
Secondly,
we find it important that
all of these models seem to
admit a
\textcolor{black}{fall into an EPM singularity via a}
smooth,
unitarity-non-violating
passage
through a strictly physical domain ${\cal D}$.
\textcolor{black}{Forcing the system to evolve into}
a
loss-of-the-observability collapse,
i.e., towards the mathematically rather specific
EPM-singularity extreme.

\textcolor{black}{Thirdly,
besides the most elementary
discrete-coordinate local-potential physics
behind model (\ref{ba0t}),
the study of the
same finite-matrix form of a realistic Hamiltonian
could also find its motivation in a
different phenomenological background.
{\it Pars pro toto}, let us mention paper \cite{chainb}
in which
Jin and Song introduced
such a
Hamiltonian (\ref{ba0t})
(in the special case with vanishing $B=C=\ldots=0$)
as a tight-binding chain
model
with possible applications
in condensed matter physics
say, of the Bloch's electronic systems with impurities.
These authors also mentioned that the same
one-parametric matrices
$H^{(N)}(A,0,0,\ldots)$
find also another,
entirely different field of applicability
in quantum information theory
where they describe the arrays of qubits.
It is also worth adding that
another immediate one-parametric generalization of the model
has been proposed by Joglekar et al \cite{chainc}.}


\begin{figure}[h]                    
\begin{center}                         
\epsfig{file=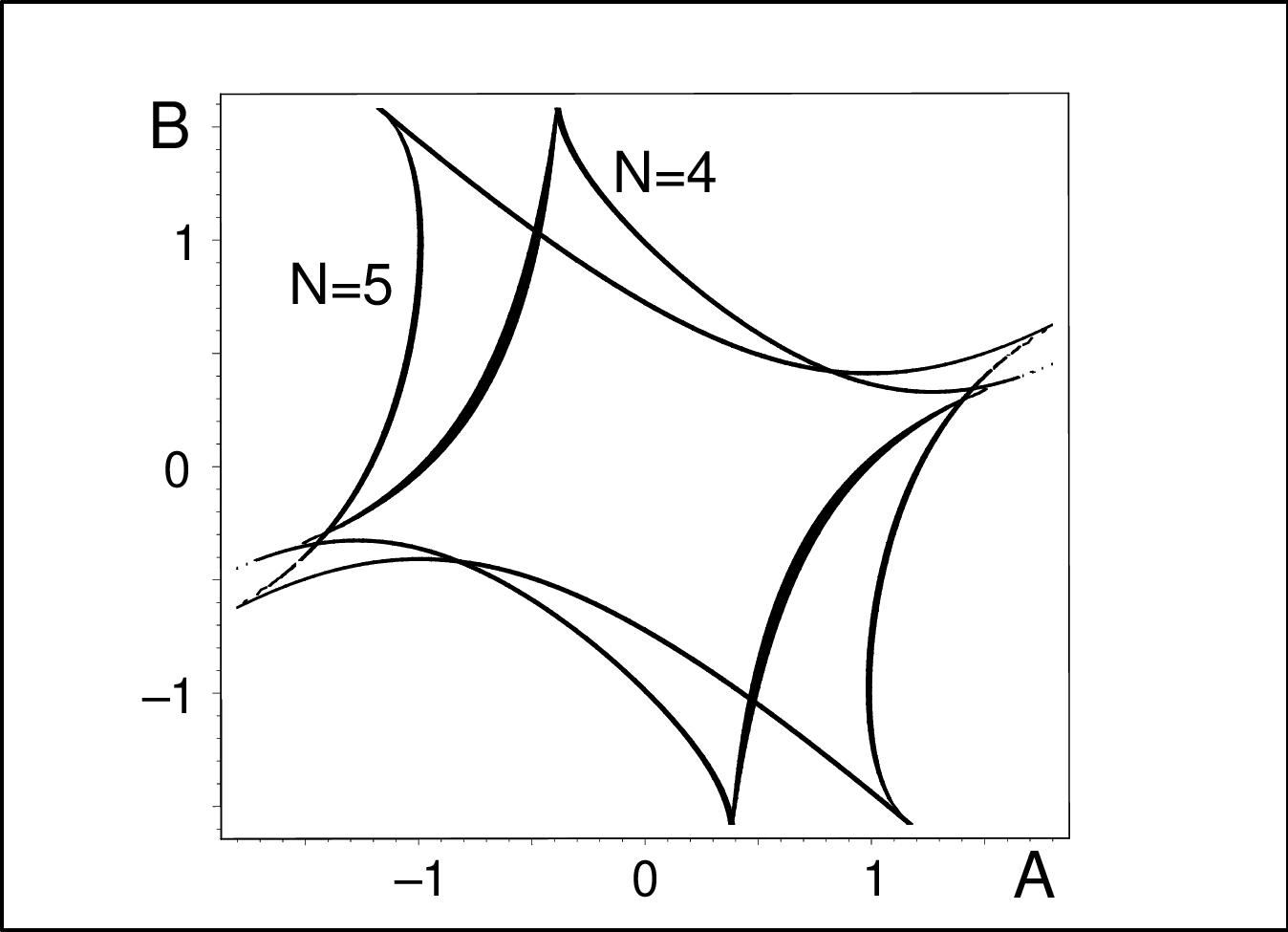,angle=270,width=0.35\textwidth}
\end{center}    
\caption{The star-shaped domains ${\cal D}$ of the reality of spectra
for the first two nontrivial
models (\ref{ba0t}) with $N=4$ and $N=5$
 \label{uglobe1}}
\end{figure}


In the virtually trivial one-parametric
version of our model with
$A \neq 0$ and with vanishing
$B=C=\ldots=0$
the localization of the central EP = EP2 singularity
is feasible, in closed form,
at an arbitrary finite matrix dimension $N$
(cf. \cite{Robin}).
In the further, nontrivial versions of the
two-parametric systems of our present interest, the
localization of the physical domains ${\cal D}={\cal D}^{(N)}$
is, in contrast, not a trivial task anymore. This is
illustrated by Figure \ref{uglobe1}
where we see a perceivable difference
between the shapes of boundaries $\partial {\cal D}^{(N)}$
at the two smallest $N$.
The picture shows that in comparison with $N=4$, the domain
of unitarity is
more protruded and
rotated to the left at $N=5$.

In the light of
the above-cited ``no-go''
observation
that
the choice of the triplet of variable parameters
$A\neq 0$, $B\neq 0$ and $C\neq 0$
leading to the extreme EPM with
$M=6$
already lies beyond
the area of practical feasibility
of the
constructions at the larger matrix dimensions $N>6$ (see \cite{jmp1}),
we are confronted with the remaining open question
concerning the feasibility of the
EPM constructions
at arbitrary $N$
in the two-parametric regime with
$A \neq 0$ and $B \neq 0$ while
$C=D=\ldots=0$.

This is the question
which is to be addressed,
and which will be affirmatively answered
in what follows.

\subsection{Secular polynomials \label{tabe}}

In an overall introduction to our forthcoming
localization of the non-Hermitian EPM degeneracies
we will consider, separately,
the two-parametric $N$ by $N$ matrix
Hamiltonians of Eq.~(\ref{ba0t})
with the even and odd $N$.
The differences are caused by the
${\cal PT}-$symmetry of the matrices
which
enables us to
evaluate the respective energy-dependent  secular
polynomials $P^{(N)}_{(A,B)}(E)$ in the
slightly different respective
forms.

\subsubsection{Even $N=2K$}

\textcolor{black}{Even a brief inspection of Figure \ref{uglobe1}
makes the differences
between the even and odd $N$s clearly visible.
Thus, once we restrict our attention, for the start,
to the even-dimensional models with $N=2K$,
we may immediately reduce the solution of
our basic eigenvalue problem
(\ref{epro})
to the search of the roots of
the related
secular polynomial
 \be
  P^{(2K)}_{(A,B)}(E)=x^K + c_1(A,B)
 x^{K-1}+ c_2(A,B)
 x^{K-2}+
 \ldots
 + c_{K-1}(A,B)x+c_{K}(A,B)\,.
 \label{mikula}
 \ee
By induction we can then immeidately
prove the following
result.}

\begin{lemma}
\label{sudalema}
At  even $N=2K=4,6,8,\ldots$,
the coefficients
in (\ref{mikula})
can be evaluated in closed form, yielding
 \be
 (-1)^K\,c_K(A,B)=(1+AB)^2-A^2\,
 \label{su6}
 \ee
and
 \be
 (-1)^K\,c_{K-1}(A,B)=B^2+K(K-1)A^2/2-(2K-1)-(K-1)(K-2)(1+AB)^2/2
 \label{su7}
 \ee
etc.
\end{lemma}
\textcolor{black}{\begin{proof}
is simplified by the $K-$independence of Eq.~(\ref{su6}).
This makes the induction-step elementary because it becomes
restricted just to the single formula of Eq.~(\ref{su7}).
\end{proof}
}


  \noindent
Due to the ${\cal PT}-$symmetry, the EPM mergers
of the $M-$plets of the energies $E_n \to E^{(EPM)}$
of our interest
occur at the very center of the spectrum, $E^{(EPM)}=0$.
This leads to our following,
easily proved key conclusion.

\begin{cor}
\label{sumalemat}
The sufficient condition of the
quadruple spectral degeneracy
at even $N$ (with $M=4$)
has the form of
coupled pair
 \be
 c_{K-1}(A,B)=0\,,
 \ \ \ \ \ \ c_{K}(A,B)=0\,
 \label{ihlarov}
 \ee
of algebraic
polynomial
equations for  $A = A^{(EPM)}$
and  $B = B^{(EPM)}$
in which the respective polynomials
as specified in
Lemma \ref{sudalema}.
\end{cor}

 \noindent
The latter \textcolor{black}{Corollary} is slightly formal because
at least some of
the solutions
 $A = A^{(EPM)}$
and  $B = B^{(EPM)}$
of
Eq.~(\ref{ihlarov})
may happen to be complex and, hence, not of our present interest.
\textcolor{black}{Secondly, in principle at least,
some of the $M-$tuple spectral degeneracies
may reflect just the presence of an exceptional point
of the order smaller than $M$ \cite{Kato}.}
An explicit classification
of these \textcolor{black}{degeneracies as well as the}
analysis of the properties
\textcolor{black}{of the energies and/or wave functions
near these boundaries of acceptability}
have to be performed case by case.
The process will be sampled in what follows.

\subsubsection{Odd $N=2K+1$}

At the odd $N=2K+1$,
obviously, it makes sense to set
 $$
 P^{(2K+1)}_{(A,B)}(E)=E\,\phi(E^2)\,.
 $$
where we can still expand
 \be
 \phi(x)=x^K + c_1(A,B)
 x^{K-1}+ c_2(A,B)
 x^{K-2}+
 \ldots
 + c_{K-1}(A,B)x+c_{K}(A,B)\,.
 \label{dmikula}
 \ee
Then, the following result can be deduced.

\begin{lemma}
\label{prvalema}
At odd
 $N=2K+1=5,7,9,\ldots$,
the coefficients
in (\ref{dmikula})
can be evaluated in closed form, yielding
 \be
 (-1)^K\,c_K=(K-1)(1+AB)^2+2-K\,A^2\,
 \label{li8}
 \ee
and
 \be
 (-1)^K\,c_{K-1}=(K-1)B^2+(K+1)K(K-1)A^2/6-K^2-K(K-1)(K-2)(1+AB)^2/6
 \label{li9}
 \ee
etc.
\end{lemma}
\textcolor{black}{\begin{proof}, by induction again,
becomes more complicated due to the explicit $K-$dependence
of both of the formulae.
Still,
no real complications emerge because we may immediately use the
tridiagonality of the matrix
and the
inspection
of the underlying secular determinant.
\end{proof}
}

\begin{cor}
\label{sulilemat}
The sufficient condition of the
quintuple spectral degeneracy at odd $N$ (with $M=5$)
has the form of
coupled pair
 \be
 c_{K-1}(A,B)=0\,,
 \ \ \ \ \ \ c_{K}(A,B)=0\,
 \label{hlarov}
 \ee
of algebraic
polynomial
equations for  $A = A^{(EPM)}$
and  $B = B^{(EPM)}$
in which the respective polynomials
as specified in Lemma \ref{prvalema}.
\end{cor}

 \noindent
\textcolor{black}{Now we are finally prepared to move from
the IEP-singular,
manifestly unphysical models with, schematically, $N=\infty$
to their EPM-singular analogues with
a suitable (and not too small) $N < \infty$,
and with just a fixed and not too large $M\leq 4$ or $M\leq 5$.}

\section{Two-parametric models
with arbitrary even $N=2K$\label{drittark}}

One of the beneficial consequences of
the not too complicated structure
of our Hamiltonian matrices (\ref{ba0t}) is that
we always have to search,
irrespectively of the parity of $N$, for the
squared-energy
roots $x=E^2$ of a polynomial
of degree $K$.
At the same time,
a deeper inspection of the problem reveals
a not quite expected structural difference between the
systems with the even $N=2K$ and with the odd $N = 2K+1$.
In and only in the latter case, for example,
there exists an odd central constant energy level $E_0=0$.

For this reason, we will study the respective two subsets
of the Hamiltonians separately.

\subsection{The next-to-elementary Hamiltonian with  $N=6$}

Once we decided to
keep our Hamiltonians just two-parametric,
the most complicated three-parametric $N=6$ model of Ref.~\cite{jmp1}
becomes simplified,
 $$
 H^{(6)}(A,B)=
 T^{(6)}+V^{(6)}(A,B)
 =
\left[ \begin {array}{cc|cc|cc}
-iA&-1&0&0&0&0\\
\noalign{\medskip}-1&-iB&-1&0&0&0\\
\hline
\noalign{\medskip}0&-1&0&-1&0&0\\
\noalign{\medskip}0&0&-1&0
&-1&0\\
\hline
\noalign{\medskip}0&0&0&-1&iB&-1\\
\noalign{\medskip}0&0&0&0&-1&
iA\end {array} \right]
\,.
 $$
It is written here in
partitioned form which emphasizes not only
the absence of the third parameter, $C=0$,
but also a certain enhancement of the symmetry
of the matrix which encourages us to recall the  $N=4$ analysis
as a methodical guide.

\textcolor{black}{At $N=6$,}
our first task is the evaluation
of the secular polynomial
 $$
 P^{(6)}_{(A,B)}(E)=
 {{\it {E}}}^{6}+ \left( -5+{A}^{2}+{B}^{2} \right) {{\it {E}}}^{4}+
 \left( 6-{B}^{2}+2\,BA-3\,{A}^{2}+{B}^{2}{A}^{2} \right) {{\it {E}}}^{
 2}-
 $$
 $$
 -1-2\,BA+{A}^{2}-{B}^{2}{A}^{2}\,.
 $$
Its real and non-degenerate
roots $E_{\pm m}\geq 0$ with $m=1,2,3$
(i.e., the sextuplet of bound-state energies)
could be written in closed form.
The formulae (easily generated using
computer-assisted symbolic manipulations)
become too long for a printed display.
Still, whenever needed,
their graphical presentation remains
instructive and straightforward (cf., e.g., \cite{ixna3}).

\textcolor{black}{In particular, it is important to notice
and emphasize that
the
purely numerically determined
star-shaped domain ${\cal D}$ of the reality of the
energy spectrum
at $N=6$
has been found to be similar and very close to its  $N=4$ predecessor
of Figure~\ref{uglobe1}.
Thus, although we are not
going to provide
a rigorous formal proof
(which
could be based on the
increase-of-precision method of paper
\cite{4a5}
and
would be, therefore, feasible),
we believe that all of the similar numerical tests seem to
confirm
a hypothesis that
the spikes of the boundaries of  ${\cal D}$
really represent the non-degenerate EPM singularities
of order four (at the even $N$s) or five (at the odd $N$s).
}

\subsection{Arbitrary $N=2K$ and the EPMs with $M=4$\label{peta} }

In what follows,
the even matrix dimensions $N=2K$ will be allowed to be
arbitry, rendering
our understanding of the behavior of the
large-matrix models possible.

\begin{figure}[h]                    
\begin{center}                         
\epsfig{file=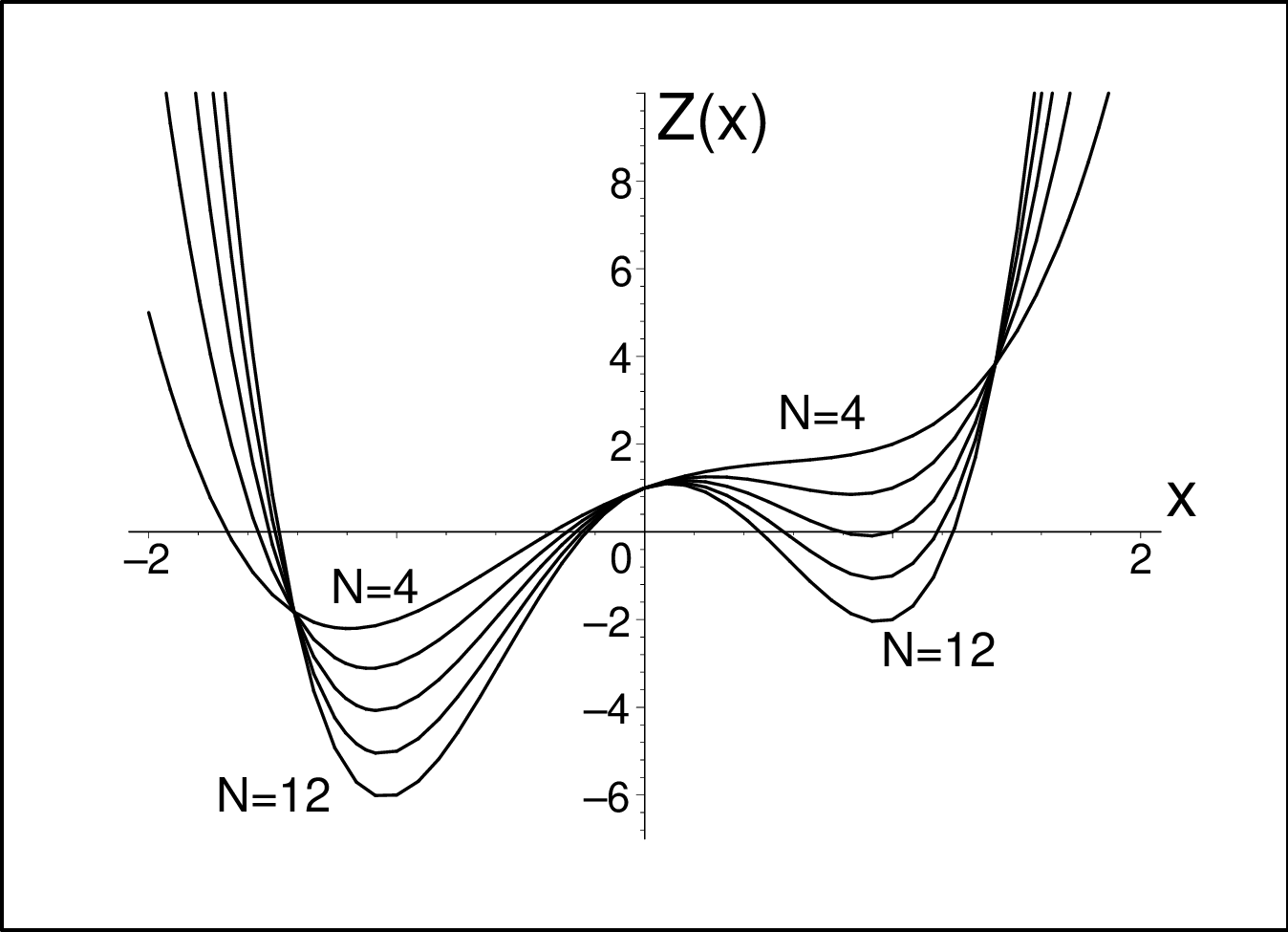,angle=270,width=0.35\textwidth}
\end{center}    
\caption{\textcolor{black}{The graphs of
polynomials $Z_{(-N)}(x)$ of Eq.~(\ref{sy11b}) with $x=x_-$, $N=2K$ and}
$K = 2,3,4,5,6$.
 \label{globe1}}
\end{figure}

\begin{thm}
\label{EP4s}
The degeneracy
 of the four central eigenvalues
 (i.e., of the four smallest eigenvalues)
 to
 $
 E^{(EP4)}(A^{(EP4)},B^{(EP4)})=0
 $
is encountered,
at any even $N=2K \geq 8$, when one
considers the pair of polynomials
 \be
 Z_{( 2K)}(x_{+})=
 (K-1)x^4_{+}-2(K-1)x^2_{+}+ 2x_{+}+1\,,
\label{sy11a}
 \ee
and
 \be
 Z_{(- 2K)}(x_{-})=
 (K-1)x^4_{-}-2(K-1)x^2_{-}- 2x_{-}+1\,,
\label{sy11b}
 \ee
and, making use of their exact solvability, when one
determines the respective quadruplets of their roots $x_+=x_{+}^{(j)}$
and $x_-=x_-^{(j)}$
with $j=1,2,3,4$.
Whenever one finds that these roots are real
(which can be shown to be true for $K \geq 4$),
the
sets of the eligible EP4-supporting parameters
in $H^{(2K)}$
become specified by the formulae
 $$
 A^{(EP4)}=A^{(j)}_{\pm}=x_{\pm}
 \,,\ \ \ \ \
 B^{(EP4)}=B^{(j)}_{\pm}=(K-1)x^3_{\pm}-2(K-1)x_{\pm}
 {\pm}1\,,
\ \ \ \ j = 1,2,3,4
 \,.
 $$
\end{thm}
\begin{proof}.
At any $N=2K$,
the coupled pair of Eqs.~(\ref{hlarov})
can immediately be solved using the
elimination
of the unknown quantity
$B$ from the second item
of Eq.~(\ref{hlarov}).
This can be done by taking the
two alternative (viz., positive or negative) square roots of
the latter constraint which may be
characterized by the subscripted sign
to yield two options, viz.,
 $$
 B=B_\pm = \pm 1-1/A\,.
 $$
After the elimination and insertion of $B_\pm= B_\pm (A) $
into the first item of Eq.~(\ref{hlarov})
we
get the respective two
polynomials (\ref{sy11a}),  (\ref{sy11b}), and the two respective
relations $Z_{(\pm 2K)}(x_{\pm})=0$.
The quadruple degeneracy of the eigenvalues at $E^{(EP4)}=0$
is guaranteed.
\end{proof}

A detailed classification based on the explicit
proofs of the reality of the roots $x_{\pm}$
of the respective polynomials $Z_{(\pm 2K)}(x_{\pm})$
remains $K$-dependent (i.e., $N$-dependent).
Even though the polynomial is of the mere fourth order in $x$
(i.e., exactly solvable),
by far the most straightforward insight in the
$K$- or $N$-dependence of the roots
is provided
graphically. In
Figure \ref{globe1}
we displayed the six shapes of
polynomial functions
$Z(x)=Z_{(-2K)}(x)$
with $K=2,3,4,5$ and $6$.
As long as both of the local
minima of these curves decrease
with the growth of $K$,
the picture clearly shows the
convergence of  the roots $x=A_-^{(EP4)}$
in the limit of $K\to \infty$.
Numerically, this convergence
is also confirmed by Table~\ref{Pwe1}.

\begin{table}[h]
\caption{Eligible EP4-supporting real roots $A=A^{(EP4)}_-$
of
polynomials $Z_{(-2K)}(x)$
and their $K \to \infty$ convergence (\textcolor{black}{incomplete
illustrative
list,
see also Eq.~(\ref{sy11b})
and Figure \ref{globe1}}).
For a completion of the list, the opposite-sign values
$A^{(EP4)}_+=-A^{(EP4)}_-$
would only have
to be added.
}
\vspace{0.21cm}
 \label{Pwe1}
\centering
\begin{tabular}{||c||c|c|c|c||}
 \hline \hline
 $K$
 &&&&\\
  \hline \hline
  2&-1.683771565& -0.3715069740
    &-
  &-
  \\
  3&-1.560602400 & -0.3133098559
  &-
  &-
  \\
  4&-1.514868938 & -0.2776482755 & 0.7925172140 & 1.000000000
    \\
  5&-1.490937129
  &-0.2527091086
  &0.5611034016
  &1.182542836
  \\
  6& -1.476207086& -0.2339114273& 0.4652866201& 1.244831894
  \\
  7&-1.466224803
  &-0.2190384827
  &0.4055099224
  &1.279753364
  \\
  \hline \hline
  $\infty$&-$\sqrt{2}$
  &0
  &0
  &$\sqrt{2}$
  \\
  \hline \hline
\end{tabular}
\end{table}

The values of $B=B^{(EP4)}_-= - 1 - 1/A^{(EP4)}_-$
are unique.
For completeness, the Table could
have been also extended to contain the
other sets of roots with opposite sign and subscript, viz., with
$A^{(EP4)}_+ =-A^{(EP4)}_-$ as well as with the
related $B^{(EP4)}_+=  1 + 1/A^{(EP4)}_+$.

\subsection{Numerical detour}

Table \ref{Pwe1} together with Figure \ref{globe1} can be recalled
as a sample of the use of a solvable model admitting an internally
consistent extrapolation of an EPM-supporting quantum model to
\textcolor{black}{the large $N \gg 1$ or even towards the limit of}
$N=\infty$. This transition can be read \textcolor{black}{as a
tentative interpolation, approximation or even
just a mere simulation}
of a
singular differential operator, with the
\textcolor{black}{idea of its possible (though not yet
sufficiently well understood)}
regularization realized via a return to \textcolor{black}{the
simplified,
non-EPM (i.e.,
diagonalizable and,
as quantum Hamiltonians,
acceptable) discrete perturbed-EPM
$N<\infty$ models.}


\textcolor{black}{
In the light of
Refs.~\cite{corridors} or \cite{photonics},
for instance, at least a few highly-excited eigenvectors of
the $N \gg 1$ model tend to overlap
at the EPM singularity. Mimicking,
in this manner, the
behavior
of the deeply singular and  manifestly unacceptable
IEP-singular model and of its
eigenvectors, with their asymptotically-high-excitation
degeneracy
as characterized by
Eq.~(\ref{aprode}) above}.
Such an \textcolor{black}{argument, still just as merely intuitive as it is at present, might}
explain
the
asymptotic (i.e., high-energy) parallelization of
eigenvectors
via its EPM-\textcolor{black}{mimicked} reinterpretation.
Simultaneously, we also circumvent,
due to the finite-dimensional form
\textcolor{black}{of the EPM models},
the very essence of the
quantum-theoretical
unacceptability of
the
\textcolor{black}{IEP-singular} operators
as sampled by $H^{(ICO)}$:

\textcolor{black}{According to the
dedicated perturbation theory
as outlined in \cite{corridors},
the reason lies
in the existence of a small-perturbation-induced
and quasi-unitarity-compatible \cite{Geyer}
regularization of the EPM
singularity
in any $N < \infty$ model. For example,}
in the light of Figure \ref{globe1}
one can expect that at large $N = 2K \gg 1$
the quickest asymptotic $K \to \infty$ convergence
will be encountered
in the case of the leftmost EP4 root
$A=A^{(EP4)}_- \approx -\sqrt{2}$
while the
slowest asymptotic $K \to \infty$ convergence
can be expected to occur for the two cental roots.

For
illustrative purposes we will
pick up, therefore,
the value ``in between''
i.e., the
rightmost
root
$A^{(EP4)}_+
\approx +\sqrt{2}$
of $Z_{(-N)}(x)$ {\it alias}, due to the left-right symmetry,
the leftmost root
$A^{(EP4)}_-
\approx -\sqrt{2}$
of the other, plus-subscripted
polynomial
 $$
Z={x}^{4}-2\,{x}^{2}+ \left( 2\,x+1 \right) g
 $$
where we changed the normalization factor, and where we
introduced
a new, asymptotically small parameter
$g = 1/({K-1})$.

We will always have  $g\leq 1$
so that
for a purely numerical localization
of any one of the four EP4 roots
we may try to use an
asymptotic-series ansatz
 \be
  x=x(K)=-\sqrt {2}+{c_1}\,g+{c_2}\,{g}^{2}+{ c_3}\,{g}^{3}+\ldots\,.
 \label{tentok}
 \ee
As long as the
very choice of our ansatz has been based on
our knowledge of
the
solution
of equation $Z=0$
in the limit of small $g \to 0$ (i.e., of large $K \to \infty$,
cf. the last line of Table \ref{Pwe1}),
the
zero-order ${\cal O}(g^0)$ form
of equation $Z=0$ is an identity.
The next,
first-order ${\cal O}(g)$ component can be read as an explicit definition
of coefficient
 $$
 c_1=-\left(4-\sqrt {2}\right) /8  \approx
 -0.3232233045\,.
 $$
Similarly, the second-order ${\cal O}(g^2)$ constraint
leads to coefficient
 $$
 c_2= \left(29\,\sqrt {2}-32 \right) / {128} \approx
 0.07040776030
 $$
while on
the third-order level of precision ${\cal O}(g^3)$ we get
 $$
 c_3=- {7}\, \left(64 -43\,\sqrt {2} \right) /{1024} \approx
 -0.02179855231
 $$
etc.

\begin{table}[h]
\caption{The sample of convergence of the asymptotic-expansion
approximants of
the leftmost
root
of
polynomial $Z_{(+N)}(x)$
at the smallest $N=2K$
(due to symmetry, we could just copy the last column from Table \ref{Pwe1}).
}
\vspace{0.21cm}
 \label{Pwe1x}
\centering
\begin{tabular}{||c||l|l|l||c||}
 \hline \hline
 $N$
 &first order&second order&third order& exact \\
  \hline \hline
   14
  & -1.466128& -1.46808&-1.4662292&-1.466224803
  \\
 12& -1.47604& -1.4788& -1.476216& -1.476207086
  \\
    10
  &-1.49062& -1.4950& -1.490959&-1.490937129
  \\
  8& -1.51413& -1.5219& -1.51494&-1.514868938
    \\
  6 & -1.5582& -1.5758&-1.56095&-1.560602400
  \\
  4&-1.6670& -1.737& -1.6888&-1.683771565
  \\
  \hline \hline
\end{tabular}
\end{table}

In this context we have to emphasize that
although equation $Z=0$
is solvable in closed form,
it still makes sense to
construct and work with
our
``redundant'' asymptotic expansion (\ref{tentok})
because
 (a) all of its coefficients can be given a closed,
exact form, and
  (b) even after its drastic truncation,
the expansion yields
a fairly reasonable numerical precision
even when $N=4$ and $g=1$ is not too small
(see Table \ref{Pwe1x}).

\section{Two-parametric models
with odd $N=2K+1 $ \label{patabe}}

\subsection{Algebraic constructions at small
$K$}

In paper \cite{jmp1}
we described, thoroughly, the odd$-N$ models with $K=1$
(i.e., the one-parametric case,
see section IV of {\it loc. cit.}) and with $K=2$
(which is the simplest two-parametric case
yielding EPM with $M=5$, see subsection V. C in {\it loc. cit.}).
This allows us to
start directly from the odd$-N$ system with $K=3$.

\subsubsection{$K=3$ and $N=7$.}

Hamiltonian
 $$
 H^{(7)}(A,B)=
 \left[ \begin {array}{ccccccc} -iA&-1&0&0&0&0&0\\
 \noalign{\medskip}-1&-iB&-1&0&0&0&0\\
 \noalign{\medskip}0&-1&0&-1&0&0&0\\\noalign{\medskip}0
&0&-1&0&-1&0&0\\\noalign{\medskip}0&0&0&-1&0&-1&0\\\noalign{\medskip}0
&0&0&0&-1&iB&-1\\\noalign{\medskip}0&0&0&0&0&-1&iA\end {array}
 \right]\,
 $$
is, for us, an eligible candidate for an energy-representing
observable in quasi-Hermitian quantum theory \cite{Geyer}
only after
we manage to specify the physical domain ${\cal D}$ of its
real parameters $A$ and $B$ for which the spectrum remains real and
non-degenerate.

The task requires the
evaluation of secular polynomial
 $$
 P^{(7)}_{(A,B)}(E)
 ={{\it {E}}}^{7}- \left( 6-{A}^{2}-{B}^{2} \right) {{\it {E}}}^{5}-
 \left( -10+2\,{B}^{2}-2\,BA+4\,{A}^{2}-{B}^{2}{A}^{2} \right) {{\it
 {E}}}^{3}-
 $$
 $$
 -\left( 4+4\,BA-3\,{A}^{2}+2\,{B}^{2}{A}^{2} \right) {\it
 {E}}.
 $$
Thus, one obvious and constant central-energy root $E_0=0$
is accompanied by the sextuplet $E_{\pm m}$ with $m=1,2,3$.
In closed form, these energies
can be expressed using the well known Cardano formulae.
Although these formulae are already too long for a display in print,
they may be stored in the computer so that
any form of the graphical
or numerical representation of the energies
still remains to be an entirely routine task.
In the same sense, it is also entirely straightforward to
find the shape of the physical parametric domain ${\cal D}$,
most easily
by the use of requirement of the reality and non-negativity of the
squares of the eigenvalues
$E_{\pm m}^2$ at $m=1,2$ and $3$.
What one obtains is just a slightly modified
analogue of Figure~\ref{uglobe1}.

\subsubsection{$K=4$ and $N=9$.}

The same
techniques apply to the
$N=9$ model yielding the secular polynomial
 $$
 P^{(9)}_{(A,B)}(E)
 =
 {{\it {E}}}^{9}- \left( 8-{A}^{2}-{B}^{2} \right) {{\it {E}}}^{7}-
 \left( -21+4\,{B}^{2}-2\,BA+6\,{A}^{2}-{B}^{2}{A}^{2} \right) {{\it
{E}}}^{5}-
 $$
 $$-\left( 20-3\,{B}^{2}+8\,BA-10\,{A}^{2}+4\,{B}^{2}{A}^{2}
 \right) {{\it {E}}}^{3}- \left( -5+4\,{A}^{2}-6\,BA-3\,{B}^{2}{A}^{2}
 \right) {\it {E}}\,.
 $$
Again, the use of the
(still existing)
closed formulae
would be impractical.
For the same reason,
also the determination of the
boundaries of the physical parametric domain ${\cal D}$
would be a purely numerical
task.

For the purposes of the description of the energy-level mergers,
fortunately,
one only has to know the boundaries of ${\cal D}$
in a small vicinity of the EPM singularity.
This makes the use
of approximate methods sufficient and efficient.


\subsection{Graphical and numerical constructions}

We have noticed that
an enormous growth of the complexity
of the formulae already
emerges at $N=2K+1$ with $K$ as small as three.
In contrast to
the above-described elementary
elimination of $B$ from Eq.~(\ref{su6}) at
an arbitrary even $N=2K$,
the odd$-N$ version of constraint $c_K(A,B)=0$ ceases
to offer a sufficiently elementary
elimination of one of the parameters. A
full-fledged computer-assisted
elimination technique
must be used to solve the system of the two
coupled polynomial algebraic equations~(\ref{hlarov}).

%
\begin{figure}[h]                    
\begin{center}                         
\epsfig{file=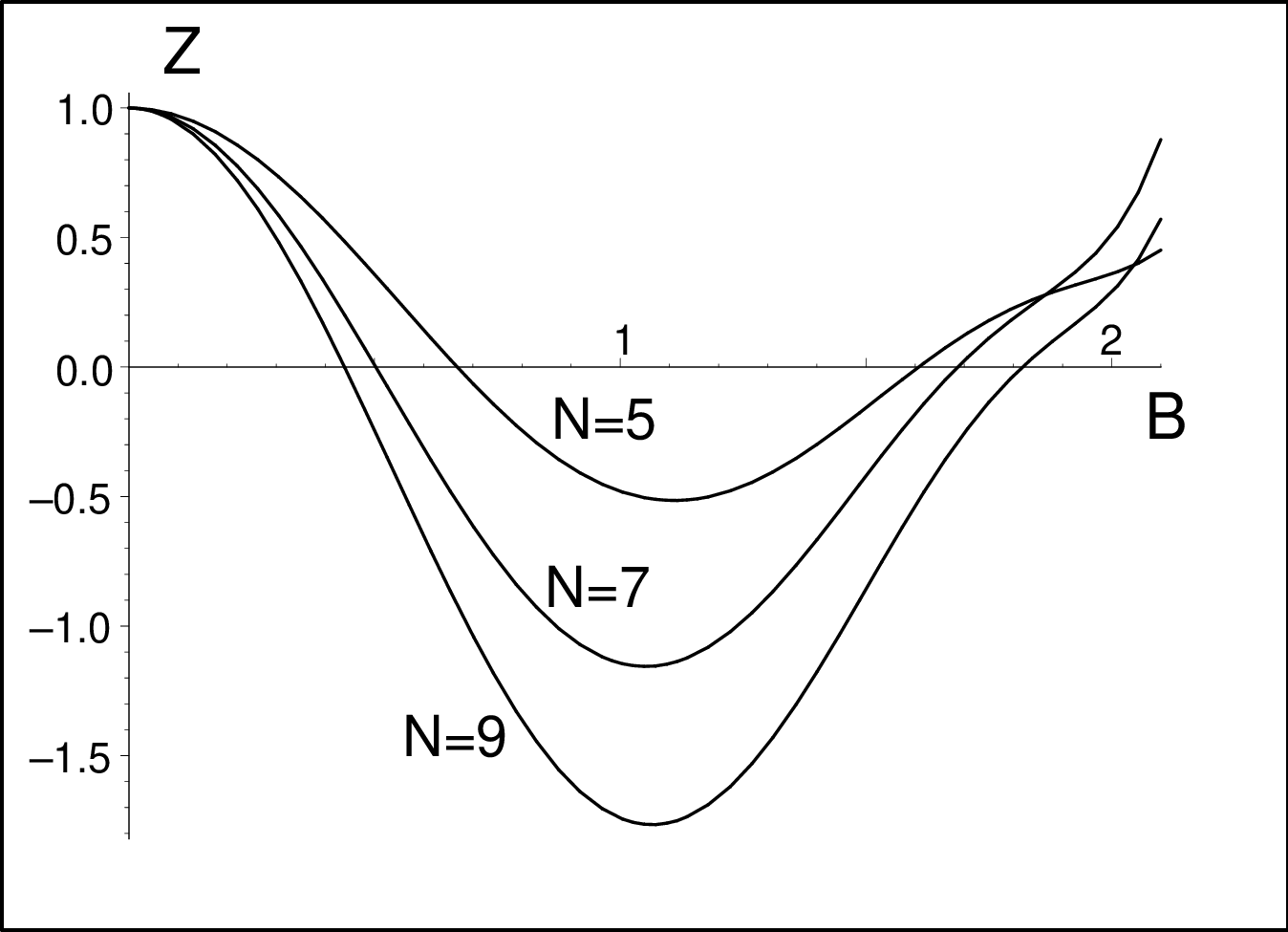,angle=270,width=0.35\textwidth}
\end{center}    
\caption{Graphical determination of the positive pairs of roots
$B=B^{(EP5)}$ of \textcolor{black}{the renormalized
polynomials $Z=Z_{(5)}(B^2)/25$,  $Z_{(7)}(B^2)/49$
and $Z_{(9)}(B^2)/900$
as defined by Eqs.~(\ref{je1}), (\ref{je2}) and (\ref{je3}), respectively}.
 \label{globe2}}
\end{figure}

Still,
we can rely upon the computer-generated
formulae.
What is encouraging is that
if one tries to formulate an odd-$N$
analogue of Theorem \ref{EP4s},
one still arrives at
certain strict analogues of
the respective auxiliary polynomials~(\ref{sy11a}) and~(\ref{sy11b})
which have to vanish at the EPM singularity.
After an abbreviation $B^2=y$ we get
the rule $Z_{(2K+1)}(y^{(EPM)})=0$
with
 \be
 Z_{(5)}(y)={y}^{4}-12\,{y}^{3}+50\,{y}^{2}-76\,{y}+25\,,
 \label{je1}
 \ee
 \be
 Z_{(7)}(y)=4\,{y}^{4}-44\,{y}^{3}+169\,{y}^{2}-234\,{y}+49\,,
 \label{je2}
 \ee
 \be
 Z_{(9)}(y)=81\,{y}^{4}-936\,{y}^{3}+3748\,{y}^{2}-5360\,{y}+900\,,
 \label{je3}
 \ee
etc. The shapes of \textcolor{black}{the latter three} curves are
displayed
in Figure~\ref{globe2},
\textcolor{black}{with the respective roots
listed in} Table~\ref{Pwe2a}.

\begin{table}[h]
\caption{Numerical zeros of functions $Z=Z_{(N)}(B^2)$ of
Figure~\ref{globe2}.
}
\vspace{0.21cm}
 \label{Pwe2a}
\centering
\begin{tabular}{||c||c|c||}
 \hline \hline
 $N$
 &\multicolumn{2}{c|}{\rm $B=B^{(EP5)}$ }\\
  \hline \hline
  5&0.6683178062 & 1.607208567
  \\
  7&0.5024794009 & 1.686679121
  \\
  9&0.4388960232 & 1.818687904
  \\
  \hline \hline
\end{tabular}
\end{table}

In contrast to the
left-right asymmetry of the curves at $N=2K$, their odd$-N$
analogues are
left-right symmetric
(that's why just their
right-semi-axis halves are displayed in Figure \ref{globe2}).
This implies that
the roots of
the new auxiliary polynomials $Z_{(2K+1)}(y)$
now define the eligible EPM parameters
as the pairs of square roots $B=B_\pm^{(EP5)} = \pm \sqrt{y}$.

The good news is that the EPM construction
is
reduced to
the search for roots of a
polynomial of the fourth order in $y$.
This
means that these roots (cf. Table \ref{Pwe2a}))
as well as the further EP5 parameters
are exact.

\begin{figure}[h]                    
\begin{center}                         
\epsfig{file=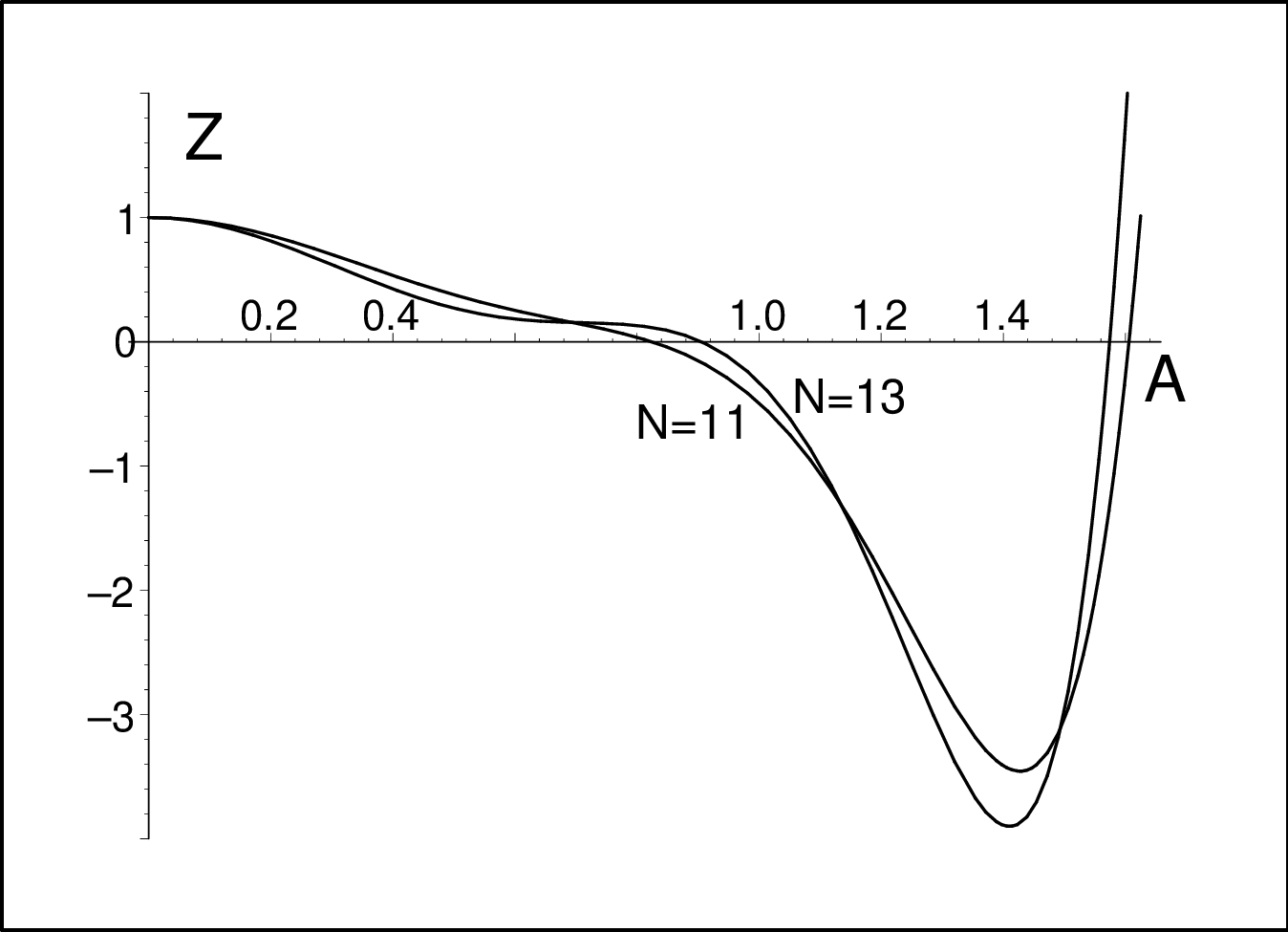,angle=270,width=0.35\textwidth}
\end{center}    
\caption{Graphical determination of the positive pairs of roots
$A=A^{(EP5)}$ of
\textcolor{black}{certain computer-generated
polynomials $Z(A)$
which are,
at $N=11$ and $13$,
not suitable for a printed display anymore
}
(see also Table~\ref{Pwe2b}).
 \label{globe3}}
\end{figure}

A few less pleasant
complications emerge
at the larger $K$s. In
contrast to the even$-N$
constructions,
the
higher-$K$ polynomials
$Z_{(2K+1)}(y)$
appear to be extrapolation-unfriendly, and they
do not
look sufficiently elementary anymore. Moreover, their
optimal forms
must be computer-generated at every separate value of $K$.

The task becomes more and more difficult with the growth of $K$.
The  computer-generated polynomials $
 Z_{(N)}(y)$
with $N = 11$ and $N=13$
cease to look nice
(so we do not display them here in print).
Even the rather routine computer-assisted numerical
localization of their EP5 roots
becomes more costly in both of its graphical forms
as sampled in Figure~\ref{globe3},
and in its purely numerical forms
as sampled in Table~\ref{Pwe2b}.

\begin{table}[h]
\caption{Numerical zeros of \textcolor{black}{the two}
functions of
Figure~\ref{globe3}.
}
\vspace{0.21cm}
 \label{Pwe2b}
\centering
\begin{tabular}{||c||c|c||}
 \hline \hline
 $N$
 &\multicolumn{2}{c||}{\rm $A=A^{(EP5)}$}\\
  \hline \hline
  11
    &0.8244776746 & 1.605982629
  \\
  13
   &0.9057668676 & 1.574311228
  \\
  \hline \hline
\end{tabular}
\end{table}


Fortunately, all of these calculations are still
routine. One can even decide to
move beyond the scope of the present paper, trying
to perform the computer-assisted
calculations at the larger $N$ and $M$.
Our readers may find a few
technical comments on such a possible future project
in Appendix~A.


\section{Summary\label{finis}}

Detailed analysis of a family of
comparatively elementary
bound-state  models
enabled us to
reveal and describe an intimate
correspondence
between the
asymptotic behavior of certain
non-Hermitian
potentials $V(x)$
and
a
non-Hermiticity-related
EP degeneracy of the
spectra.
As one  of byproducts, this led to a
deeper
understanding of the so called
intrinsic exceptional point (IEP)
feature of certain
maximally non-Hermitian but still
${\cal PT}-$symmetry-exhibiting operators.

The IEP singular behavior
has recently been noticed to
characterize a fairly broad family of ill-conceived
candidates for quantum Hamiltonians.
Even though such a property
(i.e., basically, the Riesz-basis
loss of the diagonalizability)
makes every such an operator
unacceptable as a quantum observable,
we proposed a partial remedy which lies in
a perturbation-theory based weakening if not
removal of
its singular nature.

The regularization process
has three parts.
In the first one we have
to follow the theory and we have to
circumvent the unbounded-operator
status of $H^{(IEP)}$ \cite{Geyer}.
In our \textcolor{black}{present} paper
the goal has been achieved
by the discretization of Schr\"{o}dinger equation, due to which
we could reinterpret operator $H^{(IEP)}$
as an
$N\to \infty$ limit
of a sequence of certain
$N$ by $N$ matrices $H^{(N)}$.

In the second step
we made use of the elementary matrix  nature of every $H^{(N)}$, and
we reinterpreted and generalized it as
parameter-dependent, $H^{(N)} \to H^{(N)}(g)$,
admitting that
the symbol $g$ may represent
an arbitrary multiplet of auxiliary variables, $g = \{A,B,\ldots\}$.
Even though our new matrices $H^{(N)}(g)$
forming a broader family
are, by construction, non-Hermitian,
their spectra should remain real and non-degenerate: This restricts
the variability of $g$
to a certain non-empty physical domain ${\cal D}$.

In the third step we recalled the fact that
the boundary of the latter physical parametric domain
is formed by the EP values at which one encounters
the mergers of energy levels.
A localization of these degeneracies
appeared to be the most difficult part of
our regularization recipe. Still,
in our models we managed
to localize the most relevant
part of the EP boundary $\partial {\cal D}$
at which the number $M$
of the merging energy levels
was maximal.

The latter maxima appeared to be reached at
the isolated EPs of order $M$ (EPMs),
and for the energies forming,
at a given matrix-dimension $N< \infty$, a
precise center of the spectrum.
In other words, the key point was that at every separate $N$
and in spite of
the singular nature of every
auxiliary operator $H^{(N)}(g)$ in the
$N-$ and $M-$dependent EP
limit of $g \to g^{(EPM(N))}$,
the physical unitary-evolution compatibility
of every $H^{(N)}(g)$
becomes
restored, for $g \in {\cal D}$,
in an arbitrarily small vicinity of
$g^{(EPM)}\in \partial {\cal D}$.

On these grounds it was possible to conclude that
in spite of the strictly unphysical nature of
our preselected operator $H^{(IEP)}$,
its
phenomenologically acceptable
vicinity
could be understood
as operators obtained as an $N\to \infty$
limit of their
phenomenologically acceptable $N < \infty$
predecessors $H^{(N)}(g)$,
provided only that one would be able to construct such a limit.

In our present paper we managed to perform
such a constructive regularization,
using  a subset
of matrices $H^{(N)}(g)$
which varied with
two-parametric $g=\{A,B\}$.
As a consequence,
the
finite$-N$ EPM-related central-spectral
parallelization of the
middle-of-spectrum  eigenstates of $H^{(N)}(g^{(EPM)})$
appeared successfully tractable as mimicking,
in the continuous-coordinate limit of $N \to \infty$,
the IEP-related
asymptotic
parallelization, i.e., as mimicking
one of the most relevant characteristics of the
IEP degeneracy singularity.

\textbf{Funding}
This research received no external funding.

\textbf{Data Availability Statement}
Data are contained within the article.

\textbf{Conflict of Interests}
The author declare no conflict of interest.



\section*{Appendix A. Secular polynomials}

In our recent study \cite{jmp1} of
the non-Hermitian but ${\cal PT}-$symmetric
quantum-Hamiltonian-representing
matrices $H^{(N)}(A,B,\ldots)$ of Eq.~(\ref{ba0t})
we revealed that
even after an ample use of computers,
the practical constructive determination of
the exceptional-point-representing
maxima  $A^{(ECM)}$, $B^{(ECM)}$, \ldots
of the parameters
is
complicated,
requiring a lot of computer's memory and time.

With the growth of $N$ and $M$, moreover,
the results become also
difficult to display in print.
In {\it loc. cit.},
this was the reason why
we only
described, in full detail, the simplest
two-parametric small-matrix models with $N \leq 5$
and $M \leq 5$.
Only
a sketchy note has been added
on the $N=6$
Hamiltonian with three
non-vanishing parameters
$A$, $B$ and $C$ yielding $M=6$.
The
sudden emergence of difficulties
persuaded us that
one only has to rely upon
the computer-stored results
at $N\geq 6$.

We only managed to overcome
such a rather strong skepticism
when we decided to test,
in our present paper,
the feasibility of the
constructions
in the
only remaining unexplored subset
of the models with the mere pair of parameters
$A$ and $B$ and with $M \leq 5$ and arbitrary $N$.
We succeeded, so that our skepticism
was softened.

At present
we
started
believing that at least some of the
technical obstacles
will be circumvented using some amended versions of the
Gr\"{o}bner-elimination techniques.
A practical constructive verification of such a belief is
just a challenge and open problem at present,
forming a background for a deeper study of the
solvability of the coupled polynomial algebraic equations
for the ECM maxima  $A^{(ECM)}$, $B^{(ECM)}$, \ldots
in the nearest future.

In this
setting it makes sense to list the
first few samples of secular polynomials which are
associated with matrices (\ref{ba0t}) at $N=7$ and $N=8$.

\subsection*{A.1. Three-parametric model with $N=7$}

In the case of models with $A \neq 0$, $B \neq 0$ and $C \neq 0$
and with arbitrary $N$,
our present two Lemmas and Corollary \ref{sumalemat}
would have to be properly generalized.
Using the same formal representation
of the corresponding secular polynomial as above
(cf. Eq.~(\ref{mikula})),
one would have to
replace our present fundamental algebraic set (\ref{hlarov})
by the triplet
 \be
 c_{K-2}(A,B,C)=0\,,
 \ \ \ \ \ \
 c_{K-1}(A,B,C)=0\,,
 \ \ \ \ \ \ c_{K}(A,B,C)=0\,.
 \label{ehlarov}
 \ee
We expect that such a  set of coupled polynomial
equations for  $A = A^{(EPM)}$,
$B = B^{(EPM)}$ and  $C=C^{(EPM)}$
might still be solvable
by a suitable and commercially accessible
Gr\"{o}bner elimination technique.
At the same time, the test as performed at $N=6$
in \cite{jmp1} persuaded us that
the standard Gr\"{o}bnerian reduction of the set (\ref{ehlarov})
to a single polynomial
would lead to such a high degree of this polynomial
that one could contemplate a direct numerical solution
of the coupled set (\ref{ehlarov}) itself.

For the latter purpose the choice of $N=7$
(i.e., of $K=3$)
would immediately lead to the
explicit form of relations (\ref{ehlarov}) where
one only has to insert
 $$
 c_3=
 -4+{C}^{2}-4\,BA+2\,CA+2\,{C}^{2}BA+3\,{A}^{2}+2\,CB{A}^{2}-2\,{B}^{2}
 {A}^{2}+{B}^{2}{C}^{2}{A}^{2}\,,
 $$
 $$
 c_2=
 10-2\,{B}^{2}+2\,BA-2\,{C}^{2}-4\,{A}^{2}+{B}^{2}{A}^{2}+{C}^{2}{A}^{2
 }+2\,CB+{C}^{2}{B}^{2}\,,
 $$
 $$
 c_1=
 -6+{A}^{2}+{C}^{2}+{B}^{2}\,.
 $$
A verification of the hypothesis that such an approach would
yield better results than the Gr\"{o}bner basis approach
already lies beyond the scope of the present paper.

\subsection*{A.2.  Three-parametric model with $N=8$}

In the same spirit as above, the choice of the next value of $K=4$ and
of the even $N=8$ yields
the secular polynomial with coefficients
 $$
 c_4=
1+2\,BA-{C}^{2}-2\,CA-2\,{C}^{2}BA-{A}^{2}-2\,CB{A}^{2}+{B}^{2}{A}^{2}
-{B}^{2}{C}^{2}{A}^{2}\,,
 $$
$$
 c_3=-10-6\,BA+3\,{C}^{2}-2\,CB+2\,CA+2\,{C}^{2}BA+{B}^{2}-{C}^{2}{B}^{2}+6
\,{A}^{2}+2\,CB{A}^{2}-3\,{B}^{2}{A}^{2}-{C}^{2}{A}^{2}+
$$
$$
+{B}^{2}{C}^{2}
{A}^{2}\,,
 $$
 $$
 c_2=15-3\,{B}^{2}+2\,BA-3\,{C}^{2}-5\,{A}^{2}+{B}^{2}{A}^{2}+{C}^{2}{A}^{2
}+2\,CB+{C}^{2}{B}^{2}\,,
 $$
 $$
 c_1=-7+{A}^{2}+{B}^{2}+{C}^{2}\,.
 $$
One may expect the existence of the EPM degeneracy
with $M=6$.
The guarantee and localization
of such a degeneracy
will again
require the
(computer-provided) solution of the triplet of Eqs.~(\ref{ehlarov})
so that, for the purpose, our knowledge of $c_1$ remains redundant.

\subsection*{A.3. Four-parametric model with $N=8$}

In the four-parametric models
the control of the existence of the
EPM extreme of non-Hermiticity with $M=8$
is provided
by the quadruplet of the coupled polynomial equations
 $$
 c_{K-3}(A,B,C,D)=0\,,
 \ \ \ \ \ \
 c_{K-2}(A,B,C,D)=0\,,
 $$
 \be
 \ \ \ \ \ \
 c_{K-1}(A,B,C,D)=0\,,
 \ \ \ \ \ \ c_{K}(A,B,C,D)=0\,.
 \label{eehlarov}
 \ee
At $N=8$
and $\{A,B,C,D\} \in {\cal D}$
the exact bound-state-energy roots
of secular polynomial
 $
 P^{(8)}_{(A,B,C,D)}(E)
 $
may still be defined algebraically and in closed form,
in principle at least.
In contrast, we have no estimate
concerning the computer time needed
for the solution of the EPM-determining set (\ref{eehlarov}).

At the present choice of $N=2K=8$, the separate items of
this set of equations already become too long to
be printed without abbreviations.
Thus, we will decompose $c_j=k_0+k_1\,D+k_2\,D^2$
and obtain
 $$
 c_4=1+2\,{D}^{2}CA+4\,DCBA+2\,{D}^{2}{C}^{2}BA+2\,BA-{C}^{2}-2\,CA+2\,DA+
 $$
 $$+2
 \,DC-2\,{C}^{2}BA-{A}^{2}-2\,CB{A}^{2}+{B}^{2}{A}^{2}-{B}^{2}{C}^{2}{A
 }^{2}+{D}^{2}{C}^{2}+2\,DB{A}^{2}+2\,{B}^{2}DC{A}^{2}+{B}^{2}{D}^{2}{C
 }^{2}{A}^{2}+
 $$
 $$
 +2\,{D}^{2}CB{A}^{2}+{D}^{2}{A}^{2}
 $$
i.e.,
 $$
 k_0=1+2\,BA-{C}^{2}-2\,CA-2\,{C}^{2}BA-{A}^{2}-2\,CB{A}^{2}+{B}^{2}{A}^{2}
-{B}^{2}{C}^{2}{A}^{2}\,,
 $$
 $$
 k_1=4\,CBA+2\,A+2\,C+2\,B{A}^{2}+2\,{B}^{2}C{A}^{2}
 $$
and
 $$
 k_2=2\,CA+2\,{C}^{2}BA+{C}^{2}+{B}^{2}{C}^{2}{A}^{2}+2\,CB{A}^{2}+{A}^{2}\,.
 $$
Similarly, in
 $$
 c_3=-10+2\,{D}^{2}BA+2\,{D}^{2}CB+4\,{D}^{2}-6\,BA+3\,{C}^{2}-2\,CB+2\,CA+
2\,DB-4\,DC+
$$
$$+2\,{C}^{2}BA+{B}^{2}-{C}^{2}{B}^{2}+
 $$
 $$+6\,{A}^{2}+2\,CB{A}^{2
}+{C}^{2}{D}^{2}{B}^{2}+2\,DC{B}^{2}-3\,{B}^{2}{A}^{2}-{C}^{2}{A}^{2}+
{B}^{2}{C}^{2}{A}^{2}+{B}^{2}{D}^{2}{A}^{2}+
 $$
 $$+{C}^{2}{D}^{2}{A}^{2}+2\,D
C{A}^{2}-2\,{D}^{2}{C}^{2}-2\,{D}^{2}{A}^{2}
 $$
we make use of abbreviations
 $$
 k_0=-10-6\,BA+3\,{C}^{2}-2\,CB+2\,CA+2\,{C}^{2}BA+{B}^{2}-{C}^{2}{B}^{2}+6
\,{A}^{2}+2\,CB{A}^{2}-3\,{B}^{2}{A}^{2}-{C}^{2}{A}^{2}+
$$
$$
+{B}^{2}{C}^{2}
{A}^{2}\,,
 $$
 $$
 k_1=2\,B-4\,C+2\,C{B}^{2}+2\,C{A}^{2}
 $$
and
 $$
 k_2=2\,BA+2\,CB+4+{C}^{2}{B}^{2}+{B}^{2}{A}^{2}+{C}^{2}{A}^{2}-2\,{C}^{2}-
2\,{A}^{2}\,.
 $$
Next, for
 $$
 c_2=15-3\,{B}^{2}+2\,BA-3\,{C}^{2}-4\,{D}^{2}-5\,{A}^{2}+{B}^{2}{A}^{2}+{C
}^{2}{A}^{2}+{D}^{2}{A}^{2}+
$$
$$
+2\,CB+{C}^{2}{B}^{2}+{D}^{2}{B}^{2}+2\,DC+
{D}^{2}{C}^{2}
 $$
we may set
 $$
 k_0=15-3\,{B}^{2}+2\,BA-3\,{C}^{2}-5\,{A}^{2}+{B}^{2}{A}^{2}+{C}^{2}{A}^{2
}+2\,CB+{C}^{2}{B}^{2}\,
 $$
and
 $$
 k_1=2C
 \,,\ \ \ \
 k_2=-4+{A}^{2}+{B}^{2}+{C}^{2}
 $$
while, finally, our last formula for
 $$
 c_1=-7+{A}^{2}+{B}^{2}+{C}^{2}+{D}^{2}
 $$
does not require any auxiliary abbreviations.

\end{document}